\documentclass[11pt,a4paper]{article}
\usepackage{amsmath, amsthm, graphicx, amssymb, enumerate, algorithm}

\usepackage[left=3cm, right=3cm, top=3cm]{geometry}

\usepackage{graphicx}				% Use pdf, png, jpg, or eps§ with pdflatex; use eps in DVI mode
\newtheorem{theorem}{Theorem}
\newtheorem{lemma}{Lemma}

\newtheorem{remark}{Remark}

\newcommand{\trans}{{\mathrm{T}}}

\newcommand{\E}{\operatorname{E}}
\renewcommand{\P}{\operatorname{P}}

\newcommand{\VaR}{\operatorname{V@R}}
\newcommand{\AVaR}{\operatorname{AV@R}}

\newcommand{\essinf}{\operatorname{essinf}}

\newcommand{\R}{\mathbb{R}}
\newcommand{\N}{\mathbb{N}}

\newcommand{\Var}{\operatorname{Var}}
\newcommand{\var}{\operatorname{var}}

\newcommand{\supp}{\operatorname{supp}}
\renewcommand{\phi}{\varphi}

\newcommand{\calF}{\mathcal{F}}
\newcommand{\calG}{\mathcal{G}}

\newcommand{\filF}{\mathbb{F}}
\newcommand{\filG}{\mathbb{G}}

\newcommand{\calB}{\mathcal{B}}
\newcommand{\calL}{\mathcal{L}}

\newcommand{\calP}{\mathcal{P}}

\newcommand{\cvfun}{\psi}
\newcommand{\crfun}{\nu}

\newcommand{\distr}{\stackrel{d}{\to}}

\newcommand{\weakly}{\stackrel{w}{\to}}

\title{Approximations of multi-period liability values by simple formulas}
\author{Nils Engler\footnote{nils.engler@math.su.se, Department of Mathematics, Stockholm University}, Filip Lindskog\footnote{lindskog@math.su.se, Department of Mathematics, Stockholm University}}

\begin{document}
\maketitle

\begin{abstract}
This paper is motivated by computational challenges arising in multi-period valuation in insurance. Aggregate insurance liability cashflows typically correspond to stochastic payments several years into the future. However, insurance regulation requires that capital requirements are computed for a one-year horizon, by considering cashflows during the year and end-of-year liability values. This implies that liability values must be computed recursively, backwards in time, starting from the year of the most distant liability payments. Solving such backward recursions with paper and pen is rarely possible, and numerical solutions give rise to major computational challenges. 

The aim of this paper is to provide explicit and easily computable expressions for multi-period valuations that appear as limit objects for a sequence of multi-period models that converge in terms of conditional weak convergence. Such convergence appears naturally if we consider large insurance portfolios such that the liability cashflows, appropriately centered and scaled, converge weakly as the size of the portfolio tends to infinity.  
\end{abstract}

%% keywords
\noindent {\bf Keywords}: valuation, multi-period models, conditional weak convergence

\section{Introduction and motivation}

This paper is motivated by problems arising in multi-period valuation in insurance, but the applicability of the results extends beyond the insurance setting. 
Aggregate insurance liability cashflows typically correspond to payments several years into the future. However, insurance regulation requires that capital requirements are calculated for a one-year horizon, considering cashflows during the year and end-of-year liability values.  
This implies that liability values must be computed recursively, backwards in time, starting from the year of the most distant liability payments.  
Solving such backward recursions leads to major computational challenges. Consequently, simple and explicit (standard) formulas with questionable conceptual soundness have become the preferred alternative. 
With this paper we say that there are alternatives to commonly encountered formulas that give rise to explicit formulas that are easy to use in practice and that retain both economic interpretability and the conceptual soundness of the original principles for insurance valuation. Mathematically, the explicit formulas follow from combining the original principles for insurance valuation with widely applicable large portfolio asymptotics.   

We begin by presenting the motivating insurance problem. 
Let $(C^{n_0},\filF^{n_0})$ be a stochastic processes in discrete time with $(C^{n_0}_t)_{t=1}^{T}$ denoting the discounted payments at times $t$ due to insurance claims, and $\filF^{n_0}=(\calF^{n_0}_t)_{t=0}^T$, with $\calF^{n_0}_0=\{\Omega,\emptyset\}$, denoting a filtration to which $C^{n_0}$ is adapted. $C^{n_0}$ describes discounted claim payments over time in a so-called runoff situation, when no new contracts are written. The number $n_0$ is intended as a measure of volume or exposure, for instance the number of contracts that may generate future claims payments. 
We want to assign a value to this liability cashflow. In order to do this we consider a sequence $(C^{n},\filF^n)_{n\geq 1}$ of such stochastic processes, with $\Vert C^{n} \Vert \to \infty$ almost surely as $n\to\infty$, where $\Vert \cdot \Vert$ denotes the Euclidean norm on $\R^T$.  
An example could be $C^{n}$ given by  
\begin{align*}
C^{n}_t = \sum_{k=1}^{M_n} 1_{\{D_k=t\}}Z_k, \quad t=1,\dots,T,
\end{align*}
where $M_n$ denotes the total number of claims payments and $M_n\to\infty$ almost surely as $n\to\infty$,    
$D_k$ denotes the time and $Z_k$ denotes the size of the $k$th claims payment. For this example, the filtration could be the one generated by the discounted cashflow $C^{n}$ or a larger filtration also including information about the number of payments $\sum_{k=1}^{M_n} 1_{\{D_k=t\}}$ at each point in time, or more. 

We assume nonrandom sequences with terms $a_n\in (0,\infty)$ and $b_n\in \R^T$ such that there is convergence in distribution 
\begin{align*}
X^n=a_n^{-1}(C^{n}-b_n)\distr X \quad \text{as } n\to\infty. 
\end{align*}
The value $V^n_0(C^{n})$ of the insurance liabilities in a multi-period model is the result of applying a suitable functional to the pair $(C^{n},\filF^n)$. Natural valuation functionals satisfy the property $V^n_0(C^{n})=a_nV^n_0(X^n)+\sum_{s=1}^{T}b_{n,s}$. Convergence $X^n\distr X$ suggests $V^n_0(X^n)\to V_0(X)$ and therefore the approximation 
\begin{align}\label{eq:value_approximation}
	V^n_0(C^{n})\approx a_nV_0(X)+\sum_{s=1}^{T}b_{n,s}.  
\end{align}
However, the flow of information over time is an essential ingredient in valuation and it is not true that convergence $X^n\distr X$ implies the convergence $V^n_0(X^n)\to V_0(X)$ (even if $\filF^n$ is taken to be the filtration generated by $C^{n}$). Consequently, one of the main objectives of the present paper is to determine sharp sufficient conditions for the convergence of $V^n_0(X^n)$ to $V_0(X)$ in terms of an appropriate mode of convergence of $(X^n,\filF^n)$ to $(X,\filF)$, similar to so-called extended weak convergence introduced by Aldous in \cite{Aldous-81} and conditional weak convergence studied by Sweeting in \cite{Sweeting-89}.  
The importance of the approximation \eqref{eq:value_approximation} is because $V^n_0(C^{n})$ is typically very difficult to compute numerically whereas $V_0(X)$ is easier to compute numerically and, more importantly, in the case of a Gaussian limit model, is given by an explicit expression in terms of conditional variances of components of $X$ (note that Gaussian vectors have the rare feature that conditional variances of one component given a subset of components are nonrandom).   

The paper is organized as follows. Section \ref{sec:preliminaries} introduces notation and basic properties of conditional distributions and risk measures. Section \ref{sec:compv} presents the main contents of the paper and, following a general presentation of the mathematical setup, gives economic motivation of key quantities in Section \ref{sec:economicmotivation} and presents the main results in Section \ref{sec:mainresults}. All proofs together with auxiliary results are found in Section \ref{sec:proofs}.   

\section{Preliminaries}\label{sec:preliminaries} 

$\N=\{1,2,\dots\}$, $\R$ denotes the real numbers and $\R_+=[0,\infty)$. Whenever relevant, for $T\in \N$, a vector $x\in \R^T$ is assumed to be column vector and its transpose $x^{\trans}$ a row vector. For $x=(x_1,\dots,x_T)^{\trans}\in \R^T$, ${\Vert x \Vert}^2=x^{\trans}x$ and we let $x_{\leq t}=(x_1,\dots,x_t)^{\trans}$ and $x_{> t}=(x_{t+1},\dots,x_T)^{\trans}$.  
For $d\in\N$ and a Borel set $A\in\calB(\R^d)$, $\calP(A)$ denotes the set of probability measures on $A$. For $\mu\in\calP(A)$, $\supp(\mu)$ denotes its support. 
We consider a probability space $(\Omega,\calF,\P)$. 
For a $\sigma$-algebra $\calF_t\subset \calF$, $L^0(\calF_t,\P)$ denotes the vector space of all real-valued $\calF_t$-measurable random variables, and for $p\in (0,\infty)$, $L^p(\calF_t,\P)$ denotes the subset $\{Y\in L^0(\calF_t,\P):\E[|Y|^p]<\infty\}$.  

For an $\R^d$-valued random variable $Y$ we let $\calL(Y)$ denote its distribution, i.e.~the induced probability measure $\P(Y\in \cdot)$ on $\R^d$. Given $\sigma$-algebras $\calF_t\subset \calF$ and a random variable $Y$, a regular conditional distribution $Q_{\calF_t,Y}$ is a version of $\P(Y\in \cdot\mid \calF_t)$ which forms a probability kernel from $(\Omega,\calF_t)$ to $(\R^d,\calB(\R^d))$.  
If $\calF_t=\sigma(Z)$ for an $\R^{d'}$-valued random variable $Z$, then $Q_{\calF_t,Y}$ is an $\calF_t$-measurable random measure $\omega\mapsto\kappa(Z(\omega),\cdot)$ on $\R^d$, where $\kappa$ is a probability kernel from $(\R^{d'},\calB(\R^{d'}))$ to $(\R^d,\calB(\R^d))$, see Theorem 6.3 in Kallenberg \cite{Kallenberg-02}. The notation $\calL(Y\mid Z=z)$ means $\kappa(z,\cdot)$. In particular, $z\mapsto\calL(Y\mid Z=z)$ is well defined on $\supp(\calL(Z))$.   

We define the conditional $p$-quantile of $Y\in L^0(\calF,\P)$ given $\calF_t$ as the random variable $F_{Y\mid \calF_t}^{-1}(p)$ given by  
\begin{align*}
\omega\mapsto \min\{m\in \R:Q_{\calF_t,Y}(\omega,(-\infty,m])\geq p\}. 
\end{align*}
If $Q_{\calF_t,Y}(\omega,\cdot)=\kappa(Z(\omega),\cdot)$, then $F_{Y\mid Z=z}^{-1}(p)$ means $\min\{m\in \R:\kappa(z,(-\infty,m])\geq p\}$ and is well defined on $\supp(\calL(Z))$. 
Similar to ordinary quantiles, the conditional quantiles have the property $F_{aY+b\,\mid \calF_t}^{-1}(p)=aF_{Y\mid \calF_t}^{-1}(p)+b$ if $a\in\R_+$ and $b\in L^0(\calF_t,\P)$ (in particular if $b\in\R$).  
Monetary conditional risk measures appear naturally for multi-period valuations, see e.g.~Chapter 11 in F\"ollmer and Schied \cite{foell2016}.  
The risk measures Value-at-Risk and Average Value-at-Risk, conditional on $\calF_t$ and for $u\in (0,1)$, are given by 
\begin{align*}
\VaR_{u}(Y\mid \calF_t)&=F_{-Y\mid \calF_t}^{-1}(1-u), \\
\AVaR_{u}(Y\mid \calF_t)&=\frac{1}{u}\int_0^{u}\VaR_{v}(Y\mid\calF_t)dv. 
\end{align*} 
In \cite{foell2016}, $\VaR_{u}(Y\mid \calF_t)$ is defined as $\essinf\{m_t\in L^0(\calF_t,\P):\P(Y+m_t<0\mid\calF_t)\leq u\}$ but the two definitions of $\VaR_{u}(Y\mid \calF_t)$ are equal almost surely. For $p\geq 1$, $\VaR_{u}(Y\mid \calF_t)\in L^p(\calF_t,\P)$ if $Y\in L^p(\calF,\P)$, and similarly for $\AVaR_{u}(Y\mid \calF_t)$. 
Let $\calP([0,1])'$ denote the subset of $\calP([0,1])$ consisting of measures $\mu\in\calP([0,1])$ that either have a bounded density with respect to Lebesgue measure or satisfy $\supp(\mu)\subset [a,b]$ for some $0<a<b<1$. 
Both $\VaR_{u}(Y\mid \calF_t)$ and $\AVaR_{u}(Y\mid \calF_t)$ can be expressed as 
\begin{align}\label{eq:quantriskmeas}
\rho(Y\mid \calF_t)=\int_0^1 F_{-Y\mid \calF_t}^{-1}(p)\mu(dp), \quad \mu\in \calP([0,1])'.
\end{align}
$\VaR_{u}(Y\mid \calF_t)$ corresponds to $\mu(dp)=\delta_{1-u}(dp)$ (a unit point mass at $1-u$) and $\AVaR_{u}(Y\mid \calF_t)$ corresponds to $\mu(dp)=u^{-1}I\{p\in [1-u,1]\}(dp)$ (a bounded density).  
Conditional risk measures of the form \eqref{eq:quantriskmeas} satisfy $\rho(aY+b\mid \calF_t)=a\rho(Y\mid \calF_t)-b$ if $a\in\R_+$ and $b$ is $\calF_t$-measurable, called positive homogeneity and conditional cash additivity. Moreover, they are monotone: $\widetilde{Y}\geq Y$ implies $\rho(\widetilde{Y}\mid \calF_t)\leq \rho(Y\mid \calF_t)$. Throughout the paper (in)equalities between random variables should be interpreted in the almost sure sense. 

\section{Convergence of multi-period valuations}\label{sec:compv}

Fix $T\in\N$ and let $X,X^1,X^2,\dots$ be random vectors in $\R^T$ and let $Y,Y^1,Y^2,\dots$ be random vectors in $(\R^d)^T$, $d\in\N$. 
Suppose that $\calL(X^n,Y^n)\weakly \calL(X,Y)$ as $n\to\infty$, where $\calL(X,Y)$ is Gaussian.    
Set $\calF_t=\sigma((X,Y)_{\leq t})$, $\calF^n_t=\sigma((X^n,Y^n)_{\leq t})$. $(X,Y)$ and $(X^n,Y^n)$ are adapted discrete-time stochastic processes with respect to the filtrations $(\calF_t)_{t=0}^T$ and $(\calF^n_t)_{t=0}^T$, respectively, where $\calF_0=\calF^n_0=\{\Omega,\emptyset\}$. 
For each $n$, $X^n=(X^n_t)_{t=1}^T$ corresponds to a discounted incremental cashflow in a multi-period model with $T$ periods and time points $0,1,\dots,T$. 
For each $n$, $Y^n=(Y^n_t)_{t=1}^T$ corresponds to a stochastic process that provides additional information, additional to the information provided by the process $X^n$. Taking $Y^n$ to be nonrandom means that no such additional information is considered and that $(\calF^n_t)_{t=0}^T$ is the natural filtration generated by $X^n$.  
The discounted value at time $t$ of the cashflow occurring after time $t$ is denoted by $V^n_t(X^n)$. Since no cashflows occur after time $T$, $V^n_T(X^n)=0$. The process of values of the cashflow $X^n$ is $(V^n_t(X^n))_{t=0}^T$. Both $X^n_t$ and $V^n_t(X^n)$ are $\calF^n_t$-measurable.  
The processes $(V_t(X))_{t=0}^T$ and $(V^n_t(X^n))_{t=0}^T$ are defined backward recursively by 
\begin{align}
& V_T(X)=0, \quad V_t(X)=\phi_t(X_{t+1}+V_{t+1}(X)), \quad t<T, \label{eq:VtVal}\\
& V^n_T(X^n)=0, \quad V^n_t(X^n)=\phi^n_t(X^n_{t+1}+V^{n}_{t+1}(X^n)), \label{eq:VntVal} \quad t<T, 
\end{align}
where $\phi_t : L^1(\calF_{T})\to L^1(\calF_{t})$ and $\phi^n_t : L^1(\calF^n_{T})\to L^1(\calF^n_{t})$ are mappings that are made precise below. In particular, the mappings $\phi_t$ and $\phi^n_t$ will satisfy
\begin{align}\label{eq:phitab}
\phi_t(aY+\widetilde{Y})=a\phi_t(Y)+\widetilde{Y}, \quad %\quad \text{for } a\in\R_+, b\in \R, \\
 \phi^n_t(aY+\widetilde{Y})=a\phi^n_t(Y)+\widetilde{Y} \quad \text{for } a\in\R_+, \widetilde{Y}\in L^1(\calF_t). 
\end{align}
These properties are called positive homogeneity (the effect of multiplication by a nonnegative scalar) and conditional cash additivity. 
As a consequence, for $a\in\R_+, b\in \R^T$, 
\begin{align}\label{eq:Vtab}
V_t(aX+b)=aV_t(X)+\sum_{s=t+1}^Tb_s, \quad %\text{for } a\in\R_+, b\in \R^T,\\
V^n_t(aX^n+b)=aV^n_t(X^n)+\sum_{s=t+1}^Tb_s. % \quad \text{for } a\in\R_+, b\in \R^T. 
\end{align}
Theorems \ref{thm:convergentvaluations} and \ref{thm:CondExpVar} presented below essentially say that if $\calL(X^n,Y^n)\weakly \calL(X,Y)$ as $n\to\infty$ and also in terms of conditional distributions (see e.g.~\eqref{eq:contconv}) together with a uniform integrability assumption (see e.g.~\eqref{eq:condxnui}), then $\lim_{n\to\infty}V^n_0(X^n)=V_0(X)$, where the values $V^n_0(X^n)$ and $V_0(X)$ are defined with respect to the filtrations generated by $(X^n,Y^n)$ and $(X,Y)$, respectively. The essential point is that $V^n_0(X^n)$ is typically impossible to compute analytically and only with significant difficulties numerically, whereas $V_0(X)$ has an explicit expression that is straightforward to compute in terms of the mean vector and covariance matrix of the Gaussian weak limit (see e.g.~\eqref{eq:V0expression}). 
Therefore, the theorems we present enable the use of conceptually sound valuation techniques without the significant efforts otherwise needed to obtain numerical solutions to backward recursions. 

\subsection{Economic motivation}\label{sec:economicmotivation}

Current regulatory frameworks for the insurance industry prescribe so-called cost-of-capital valuation of liability cashflows. 
Such approaches to valuation consider capital requirements and the costs stemming from the financing of buffer capital. In the multi-period setting, with capital requirements determined one period at the time, the randomness of future capital requirements and associated costs drive the liability valuation. 

Multiperiod cost-of-capital valuation is studied in 
Salzmann and W\"uthrich \cite{SalzmannWuthrich-10}, 
M\"ohr \cite{mhr11},  
Pelsser and Salahnejhad Ghalehjooghi \cite{Pelsser-SG-16} 
and  
Engsner et al. \cite{engsner2017}. 
A common theme is that multiperiod valuations is constructed through backward induction of one-period valuations. 
Pelsser and Salahnejhad Ghalehjooghi \cite{Pelsser-SG-16} study continuous-time limits of multiperiod valuations defined in terms of one-step valuations, similar to those considered below and including cost-of-capital valuation. In \cite{Pelsser-SG-16} the convergence takes place as the length of the time periods tends to zero and the number of time periods tends to infinity. Clearly, this is a different kind of convergence than that studied here. 
Multiperiod valuations are so-called time-consistent by their construction through the backward induction of one-period valuations. For mathematical properties of general multiperiod valuations, we refer to 
Cheridito et al \cite{Cheridito-et-al-06}, Artzner et al. \cite{Artzner-et-al-07} and Jobert and Rogers \cite{JobertRogers-08}.  

We now explain the basic ingredients of multi-period cost-of-capital valuation. 
Let $V_0$ be the value at time $0$ of the liability cashflow $X$. This amount should be interpreted as the capital that needs to be transferred along with the liability to another external agent in order for the external agent (or capital provider) to accept managing the liability runoff and the associated capital costs.  
Let $V_{t}$ denote the value at time $t$ of the liability cashflows $X_{>t}$. The capital $V_{t}$ is reserved at time $t$ for managing the liability. However, regulation requires the capital $R_t=\rho(-X_{t+1}-V_{t+1}\mid \calF_t)>V_{t}$ to be set aside at time $t$. The difference $R_t-V_t$ is made available by a capital provider requiring an excess expected return $1+\eta_t$ on the provided capital between time $t$ and $t+1$. The acceptability criterion under which the capital provider accepts to provide capital gives the equation 
\begin{align}\label{eq:accr_CoCnoLL}
\E[R_t-X_{t+1}-V_{t+1}\mid \calF_t]=(1+\eta_t)(R_t-V_t). 
\end{align}  
Solving for $V_t$ yields $V_t=\phi_t(X_{t+1}+V_{t+1})$, where 
\begin{align}\label{eq:phiCoCnoLL}
\phi_t(Y)=\frac{1}{1+\eta_t}\E[Y\mid \calF_t]+\frac{\eta_t}{1+\eta_t}\rho(-Y\mid \calF_t).
\end{align}
A commonly used conditional risk measure for variables with finite variance is $\rho^{\text{SD}}(Y\mid \calF_t)=\E[-Y\mid \calF_t]+c^{\text{SD}}\var(Y\mid \calF_t)^{1/2}$. Note that it can be argued that $\rho^{\text{SD}}(Y\mid \calF_t)$ is an inappropriate choice because it violates the monotonicity property: $\widetilde{Y}\geq Y$ does not imply $\rho^{\text{SD}}(Y\mid \calF_t)\geq \rho^{\text{SD}}(\widetilde{Y}\mid \calF_t)$. For this choice of conditional risk measure, \eqref{eq:phiCoCnoLL} takes the form
\begin{align}\label{eq:phi_esd}
\phi_t(Y)=\E[Y\mid \calF_t]+c_t\var(Y\mid \calF_t)^{1/2},
\end{align}
where $c_t=c^{\text{SD}}\eta_t/(1-\eta_t)$, and $\phi_t(Y)\in L^2(\calF_t)$ if $Y\in L^2(\calF_T)$.  

It can be argued that the capital provider, seen as the share holder of the company, has limited liability and is not required to continue injecting capital if the the value $R_t$ of available asset turns out insufficient to match the value $X_{t+1}+V_{t+1}$ of the liability towards the policy holders. 
For discussions on limited liability in the context of cost-of-capital valuation we refer to Albrecher et al. \cite{Albrecher-et-al-22} and M\"ohr \cite{mhr11}.  
In the setting with limited liability, \eqref{eq:accr_CoCnoLL} is replaced by 
\begin{align}\label{eq:accr_CoCLL}
\E[(R_t-X_{t+1}-V_{t+1})^{+}\mid \calF_t]=(1+\eta_t)(R_t-V_t)
\end{align}    
and \eqref{eq:phiCoCnoLL} is replaced by 
\begin{align}\label{eq:phiCoCLL}
\phi_t(Y)=\rho(-Y\mid \calF_t)-\frac{1}{1+\eta_t}\E[(\rho(-Y\mid\calF_t)-Y)^{+}\mid \calF_t]. 
\end{align}
Regardless of whether we consider cost-of-capital valuation with or without limited liability, it follows that $\phi_t(aY+\widetilde{Y})=a\phi_t(Y)+\widetilde{Y}$ for $a\in\R_+$ and $\widetilde{Y}\in L^1(\calF_t)$ if the conditional risk measure satisfies $\rho(aY+\widetilde{Y}\mid \calF_t)=a\rho(Y\mid \calF_t)-\widetilde{Y}$. Moreover, $\phi_t$ inherits monotonicity from the conditional risk measure: if $\widetilde{Y}\geq Y$ implies $\rho(Y\mid \calF_t)\geq \rho(\widetilde{Y}\mid \calF_t)$, then $\phi_t(\widetilde{Y})\geq \phi_t(Y)$. 

An alternative to an acceptability criterion based on expected excess return and a cost-of-capital rate is an acceptability criterion saying that a risk averse capital provider provides capital if the payoff resulting from providing capital is preferred, in terms of expected utility, to simply rolling this capital forward by investing it in a riskless bond (or more generally, investing it in the numeraire asset). In this setting, the acceptability criterion \eqref{eq:accr_CoCLL} is replaced by 
\begin{align}\label{eq:accr_EULL}
\E[u_t((R_t-X_{t+1}-V_{t+1})^{+})\mid \calF_t]=u_t(R_t-V_t), 
\end{align}    
where $u_t$ is an increasing and concave (utility) function.  
Consequently, \eqref{eq:phiCoCLL} is replaced by 
\begin{align}\label{eq:phiEULL}
\phi_t(Y)=\rho(-Y\mid \calF_t)-u_t^{-1}\Big(\E\Big[u_t((\rho(-Y\mid\calF_t)-Y)^{+})\mid\calF_t\Big]\Big).
\end{align}
Also in this case $\phi_t$ inherits monotonicity and the property $\phi_t(Y+\widetilde{Y})=\phi_t(Y)+\widetilde{Y}$, $\widetilde{Y}\in L^1(\calF_t)$, from a monotone and conditionally cash additive conditional risk measure. However, the property $\phi_t(aY+\widetilde{Y})=a\phi_t(Y)+\widetilde{Y}$, $a\in\R_+$, requires that $u_t$ is chosen a power utility function $u_t(x)=\alpha_t x^{\beta_t}$, where $\alpha_t>0$ and $\beta_t\in (0,1]$. In this case, $\phi_t$ in \eqref{eq:phiEULL} takes the form
\begin{align}\label{eq:phiPULL}
\phi_t(Y)=\rho(-Y\mid \calF_t)-\E\Big[\big((\rho(-Y\mid\calF_t)-Y)^{+}\big)^{\beta_t}\mid\calF_t\Big]^{1/\beta_t}.
\end{align} 
The mappings $\phi_t:L^1(\calF_T)\to L^1(\calF_t)$ given by \eqref{eq:phiCoCnoLL}, \eqref{eq:phiCoCLL} or \eqref{eq:phiPULL} all satisfy $\phi_t(Y+\widetilde{Y})=\phi_t(Y)+\widetilde{Y}$ whenever $\widetilde{Y}\in L^1(\calF_t)$. Therefore, \eqref{eq:VtVal} can be expressed as, with $\circ$ denoting composition, 
\begin{align*}
V_T(X)=0, \quad V_t(X)=\phi_t \circ \dots \circ \phi_{T-1}\bigg(\sum_{s=t+1}^{T}X_s\bigg), \quad t < T.
\end{align*} 
These properties also hold for the mappings $\phi_t$ in \eqref{eq:phi_esd} with $L^1$ replaced by $L^2$. 

\subsection{Main results}\label{sec:mainresults}

The main results consist of Theorems \ref{thm:convergentvaluations}, \ref{thm:CondExpVar} and \ref{thm:V0monotonicity}. 
Theorems \ref{thm:convergentvaluations} and \ref{thm:CondExpVar} presents conditions under which we have convergence for multi-period values (that are typically not computable) to a computable explicit limit. Theorem \ref{thm:V0monotonicity} presents a monotonicity result for the limit expressions in terms of a partial order between filtrations.  

\begin{theorem}\label{thm:convergentvaluations}
Let $X,X^1,X^2,\dots$ be random vectors in $\R^T$ and let $Y,Y^1,Y^2,\dots$ be random vectors in $(\R^d)^T$, $d\in \N$. 
Suppose that $\calL(X^n,Y^n)\weakly \calL(X,Y)$ as $n\to\infty$, where $\calL(X,Y)$ is Gaussian, and that $(\calL(\Vert X^n\Vert))_{n\in\N}$ is uniformly integrable.  
Suppose also that, for each $t$ and each convergent sequence $(x^n,y^n)_{\leq t}\to (x,y)_{\leq t}$ with $(x^n,y^n)_{\leq t}\in \supp(\calL((X^n,Y^n)_{\leq t}))$,  
\begin{align}
& \calL\Big((X^n,Y^n) \mid (X^n,Y^n)_{\leq t}=(x^n,y^n)_{\leq t}\Big) \weakly \calL\Big((X,Y) \mid (X,Y)_{\leq t}=(x,y)_{\leq t}\Big) 
\quad \text{as } n\to\infty \label{eq:contconv}, \\ 
& \Big(\calL\big(\Vert X^n\Vert \mid (X^n,Y^n)_{\leq t}=(x^n,y^n)_{\leq t}\big)\Big)_{n\in\N} \quad \text{is uniformly integrable}. \label{eq:condxnui} 
\end{align}
For all $t$, let $V_t(X)$ and $V^n_t(X^n)$ be given by \eqref{eq:VtVal} and \eqref{eq:VntVal} with $\phi_t$ and $\phi^n_t$ given by either 
\eqref{eq:phiCoCnoLL} and 
\begin{align}\label{eq:phiCoCnoLL_n}
\phi^n_t(Y)=\frac{1}{1+\eta_t}\E[Y\mid \calF^n_t]+\frac{\eta_t}{1+\eta_t}\rho(-Y\mid \calF^n_t), \quad \eta_t \in \R_+,
\end{align}
or by \eqref{eq:phiCoCLL} and 
\begin{align}\label{eq:phiCoCLL_n}
\phi^n_t(Y)=\rho(-Y\mid \calF^n_t)-\frac{1}{1+\eta_t}\E\Big[\big(\rho(-Y\mid\calF^n_t)-Y\big)^{+}\mid \calF^n_t\Big], \quad \eta_t \in \R_+,
\end{align}
or by \eqref{eq:phiPULL} and 
\begin{align}\label{eq:phiPULL_n}
\phi^n_t(Y)=\rho(-Y\mid \calF^n_t)-\E\Big[\big((\rho(-Y\mid\calF^n_t)-Y)^{+}\big)^{\beta_t}\mid\calF^n_t\Big]^{1/\beta_t}, \quad \beta_t\in (0,1],
\end{align} 
where $\calF_t=\sigma((X,Y)_{\leq t})$ and $\calF^n_t=\sigma((X^n,Y^n)_{\leq t})$, and where $\rho(-Y\mid \calF_t)$ and $\rho(-Y\mid \calF^n_t)$ are of the form \eqref{eq:quantriskmeas}. 
Then $\lim_{n\to\infty}V^n_0(X^n)=V_0(X)$, where 
\begin{align}\label{eq:V0expression}
V_0(X)=\E\bigg[\sum_{t=1}^T X_t\bigg]+\sum_{t=1}^T\phi_{t-1}(\varepsilon_t)\bigg(\Var\bigg(\sum_{u=t}^TX_u\mid (X,Y)_{\leq t-1}\bigg)
-\Var\bigg(\sum_{u=t}^TX_u\mid (X,Y)_{\leq t}\bigg)\bigg)^{1/2},
\end{align}
where $\varepsilon_t$ is standard normally distributed and independent of $\calF_{t-1}$. 
\end{theorem}

\begin{remark}
Notice that the conditional variances in \eqref{eq:V0expression} are nonrandom since the (joint) distribution of $(X,Y)$ is Gaussian. 
Notice also $\calF_t=\sigma((X,Y)_{\leq t})$ generates the natural filtration of $X$ if $Y$ is chosen as nonrandom (a degenerate Gaussian process). We emphasize that if the $\eta_t$ in \eqref{eq:phiCoCnoLL} and \eqref{eq:phiCoCLL} do not depend on $t$, and similarly if the $\beta_t$ in \eqref{eq:phiPULL_n} do not depend on $t$, then $\phi_{t-1}(\varepsilon_t)=\phi_0(\varepsilon_1)$ in \eqref{eq:V0expression} does not depend on $t$. 
\end{remark}

\begin{remark}
Joint weak convergence $\calL(X^n,Y^n)\weakly \calL(X,Y)$ as $n\to\infty$ does not imply conditional weak convergence such as the convergence in \eqref{eq:contconv}. A counterexample as well as sufficient conditions for conditional weak convergence are presented in \cite{Sweeting-89}. 
\end{remark}

\begin{remark}
For $\varepsilon$ is standard normally distributed and independent of $\calF_{t-1}$, 
\begin{align*}
\VaR_{u}(\varepsilon \mid \calF_{t-1}) = \Phi^{-1}(1-u), \quad \AVaR_{u}(\varepsilon \mid \calF_{t-1}) = \frac{1}{u}\phi\left(\Phi^{-1}(1-u)\right),
\end{align*}
where $\phi$ and $\Phi$ here denote the standard normal density and distribution function, respectively. Hence, with $\rho_0(\varepsilon)$ denoting either $\VaR_{u}(\varepsilon \mid \calF_{t-1})$ or $\AVaR_{u}(\varepsilon \mid \calF_{t-1})$ with values above that do not depend on $t$, we get 
$$
\phi_{t-1}(\varepsilon) = \frac{\eta_{t-1}}{1 + \eta_{t-1}} \rho_0(\varepsilon)
$$
in the case of \eqref{eq:phiCoCnoLL}, and 
\begin{align*}
\phi_{t-1}(\varepsilon) = \rho_0(\varepsilon) - \frac{1}{1 + \eta_{t-1}}\Big(\rho_0(\varepsilon)\Phi(\rho_0(\varepsilon)) + \phi(\rho_0(\varepsilon)) \Big)\leq \frac{\eta_{t-1}}{1 + \eta_{t-1}} \rho_0(\varepsilon) 
\end{align*}
in the case of \eqref{eq:phiCoCLL}. In case of \eqref{eq:phiPULL}, the expectation $\E[((\rho_0(\varepsilon)-\varepsilon)^+)^{\beta_{t-1}}]$ has to be computed numerically in order to compute $\phi_{t-1}(\varepsilon)$.   
\end{remark}

\begin{remark}\label{rem:noll}
As stated in Theorem \ref{thm:convergentvaluations}, the result holds when $\phi_t$ and $\phi^n_t$ are given by \eqref{eq:phiCoCnoLL} and \eqref{eq:phiCoCnoLL_n}. However, the theorem is actually proven for mappings of a more general kind, namely, for $\lambda_t\in [0,1]$ and $\mu^1_t,\mu^2_t\in \calP([0,1])'$, the mappings $\phi_t : L^1(\calF_{T})\to L^1(\calF_{t})$, $\phi^n_t : L^1(\calF^n_{T})\to L^1(\calF^n_{t})$  given by \eqref{eq:phialtmapping} and \eqref{eq:phinaltmapping} below:    
\begin{align}
\phi_t(Y)&=\lambda_t\int F^{-1}_{Y \mid \calF_t}(p)\mu_t^1(dp)+(1-\lambda_t)\int F^{-1}_{-Y \mid \calF_t}(p)\mu_t^2(dp), \label{eq:phialtmapping}\\
\phi^n_t(Y)&=\lambda_t\int F^{-1}_{Y \mid \calF^n_t}(p)\mu_t^1(dp)+(1-\lambda_t)\int F^{-1}_{-Y \mid \calF^n_t}(p)\mu_t^2(dp).   \label{eq:phinaltmapping}
\end{align}
Note that \eqref{eq:phiCoCnoLL} and \eqref{eq:phiCoCnoLL_n} correspond to $\mu_t^1(dp)=dp$ and $\lambda_t=(1+\eta_t)^{-1}$. 
\end{remark}

\begin{remark}\label{rem:ll}
Both \eqref{eq:phiCoCLL}, \eqref{eq:phiCoCLL_n} and \eqref{eq:phiPULL}, \eqref{eq:phiPULL_n} are special cases of 
\begin{align}
\phi_t(Y)&=\rho(-Y\mid \calF_t)-\gamma_t\E\Big[\big((\rho(-Y\mid\calF_t)-Y)^{+}\big)^{\beta_t}\mid\calF_t\Big]^{1/\beta_t}, \label{eq:phiPULLgen}\\
\phi^n_t(Y)&=\rho(-Y\mid \calF^n_t)-\gamma_t\E\Big[\big((\rho(-Y\mid\calF^n_t)-Y)^{+}\big)^{\beta_t}\mid\calF^n_t\Big]^{1/\beta_t}, \label{eq:phiPULLgen_n}
\end{align}
as seen by choosing $(\gamma_t,\beta_t)=((1+\eta_t)^{-1},1)$ in case of   \eqref{eq:phiCoCLL}, \eqref{eq:phiCoCLL_n}, and $\gamma_t=1$ in case of \eqref{eq:phiPULL}, \eqref{eq:phiPULL_n}.  
\end{remark}

Our next result shows convergence of multi-period values in the case of where the one period valuation mapping is defined as a sum of a conditional expectation and a constant times a conditional standard deviation. 
Let $c_t \in \R_+$ and
$\phi_t : L^2(\calF_{T})\to L^2(\calF_{t})$, $\phi^n_t : L^2(\calF^n_{T})\to L^2(\calF^n_{t})$ be given by
\begin{align}
\phi_t(Y)&=\E \big[Y \mid \calF_t \big] + c_t \var \big[Y \mid \calF_t \big]^{1/2}, \label{eq:phimappingII} \\
\phi^n_t(Y)&=\E \big[Y \mid \calF^n_t \big] + c_t \var \big[Y \mid \calF^n_t \big]^{1/2}. \label{eq:phinmappingII} 
\end{align}
$(V_t(X))_{t=0}^T$ and $(V^n_t(X^n))_{t=0}^T$ are defined backward recursively by \eqref{eq:VtVal} and \eqref{eq:VntVal}, and \eqref{eq:phitab} and \eqref{eq:Vtab} hold.  

\begin{theorem}\label{thm:CondExpVar}
Let $X,X^1,X^2,\dots$ be random vectors in $\R^T$ and let $Y,Y^1,Y^2,\dots$ be random vectors in $(\R^d)^T$, $d\in\N$. Suppose that $\calL(X^n,Y^n)\weakly \calL(X,Y)$ as $n\to\infty$, where $\calL(X,Y)$ is Gaussian, and such that $(\calL(\Vert X^n\Vert^2))_{n\in\N}$ is uniformly integrable. Suppose also that, for each $t$ and each convergent sequence $(x^n,y^n)_{\leq t}\to (x,y)_{\leq t}$ with $(x^n,y^n)_{\leq t}\in \supp(\calL((X^n,Y^n)_{\leq t}))$,     
\begin{align}
& \calL\Big((X^n,Y^n) \mid (X^n,Y^n)_{\leq t}=(x^n,y^n)_{\leq t}\Big) \weakly \calL\Big((X,Y) \mid (X,Y)_{\leq t}=(x,y)_{\leq t}\Big) \quad \text{as } n\to\infty, \label{eq:contconvII} \\
& \Big(\calL(\Vert X^n\Vert^2 \mid (X^n,Y^n)_{\leq t}=(x^n,y^n)_{\leq t})\Big)_{n\in\N} \quad \text{is uniformly integrable}. \label{eq:condxnuiII} 
\end{align}
Let $V_t(X)$ and $V^n_t(X^n)$ be given by \eqref{eq:VtVal} and \eqref{eq:VntVal} with $\phi_t$ and $\phi^n_t$ given by \eqref{eq:phimappingII} and \eqref{eq:phinmappingII} with $\calF_t=\sigma((X,Y)_{\leq t})$ and $\calF^n_t=\sigma((X^n,Y^n)_{\leq t})$. 
Then $\lim_{n\to\infty}V^n_0(X^n)=V_0(X)$, where 
\begin{align*}
%\begin{split}
V_0(X)=\E\bigg[\sum_{t=1}^T X_t\bigg]+\sum_{t=1}^T\phi_{t-1}(\varepsilon_t)\bigg(\Var\bigg(\sum_{u=t}^TX_u\mid (X,Y)_{\leq t-1}\bigg)
-\Var\bigg(\sum_{u=t}^TX_u\mid (X,Y)_{\leq t}\bigg)\bigg)^{1/2},
%\end{split}
\end{align*}
where $\varepsilon_t$ is standard normally distributed and independent of $\calF_{t-1}$. 
\end{theorem}

The final theorem considers the Gaussian limit models appearing in Theorems \ref{thm:convergentvaluations} and \ref{thm:CondExpVar} and establishes a partial ordering of values $V_0(X)$ assigned to a given cashflow $X$ depending on a partial order between filtrations to which $X$ is adapted. One may guess that a larger filtration allows for more accurate predictions over time of future cashflows and thereby reduces the value $V_0(X)$. We show that this holds if the parameters of the one-step valuation mappings $\phi_t$ do not depend on $t$ and if the conditional variance of the aggregate cashflow $\sum_{s=1}^TX_s$ decays in a convex manner over time as more information becomes available.  

We consider filtrations $\filF=(\calF_t)_{t=0}^{T}$ with $\calF_0=\{\Omega,\emptyset\}$ and $\calF_t=\sigma((X,Y)_{\leq t})$, for $t\geq 1$, where $(X,Y)$ is jointly Gaussian. We assume that the one-step valuation mappings $\phi_t$ do not depend on $t$ such that $\phi_{t-1}(\varepsilon_t)=\phi_0(\varepsilon_1)$ in \eqref{eq:V0expression}. We write 
\begin{align*}
V_0(X,\filF)=\E\bigg[\sum_{t=1}^T X_t\bigg]+\phi_{0}(\varepsilon_1)\sum_{t=1}^T\bigg(&\Var\bigg(\sum_{u=t}^TX_u\mid \calF_{t-1}\bigg)
-\Var\bigg(\sum_{u=t}^TX_u\mid \calF_{t}\bigg)\bigg)^{1/2}.
\end{align*}
The following theorem says that for filtrations $\filF$ and $\filG$ of the above kind (generated by a process that is jointly Gaussian together with the cashflow process $X$), if $t\mapsto \Var(\sum_{s=1}^TX_s \mid \calF_t)$ is convex, then $V_0(X,\filF)\geq V_0(X,\filG)$ whenever $\calF_t\subseteq \calG_t$ for every $t$. The convexity assumption is necessary.  

\begin{theorem}\label{thm:V0monotonicity}
Let $(X,Y)_{t=1}^T$ and $(X,Z)_{t=1}^T$ be Gaussian processes. 
Let $\filF=(\calF_t)_{t=0}^{T}$ and $\filG=(\calG_t)_{t=0}^{T}$ be filtrations with $\calF_0=\calG_0=\{\Omega,\emptyset\}$ and $\calF_t=\sigma((X,Y)_{\leq t})$ and $\calG_t=\sigma((X,Z)_{\leq t})$, for $t\geq 1$. 
If $\calF_t\subseteq \calG_t$ for every $t$ and if $t\mapsto \Var(\sum_{s=1}^TX_s \mid \calF_t)$ is convex, then $V_0(X,\filF)\geq V_0(X,\filG)$. 
\end{theorem}

\begin{remark}
We emphasize that the Gaussian assumption in Theorem \ref{thm:V0monotonicity} means that the terms  
\begin{align*}
c_t=\Var\bigg(\sum_{s=1}^TX_s \mid \calF_{t-1}\bigg)-\Var\bigg(\sum_{s=1}^TX_s \mid \calF_{t}\bigg) 
\end{align*}
form a nonrandom sequence $(c_t)_{t=1}^{T}$ and that the convexity property is equivalent to 
$c_1\geq c_2 \geq \dots \geq c_{T}$. 
\end{remark}

\section{Proofs and auxiliary results}\label{sec:proofs}

We use the notation $Z^n_{\leq t}=(X^n,Y^n)_{\leq t}$, $Z_{\leq t}=(X,Y)_{\leq t}$ and $z^n_{\leq t}=(x^n,y^n)_{\leq t}$, $z_{\leq t}=(x,y)_{\leq t}$.  

\begin{proof}[Proof of Theorem \ref{thm:convergentvaluations} assuming \eqref{eq:phialtmapping} and \eqref{eq:phinaltmapping}]
As explained in Remark \ref{rem:noll}, \eqref{eq:phiCoCnoLL} and \eqref{eq:phiCoCnoLL_n} are special cases of \eqref{eq:phialtmapping} and \eqref{eq:phinaltmapping}. 

For $t=1,\dots,T$, consider mappings $\cvfun^n_t:\supp(\calL(Z^n_{\leq t}))\to \R$ given by $\cvfun^n_T(z^n_{\leq T})=\sum_{s=1}^Tx^n_s$ and 
\begin{align}\label{eq:cvfun_ta}
\cvfun^n_t(z^n_{\leq t}) 
=\lambda_t\int F^{-1}_{\cvfun^n_{t+1}(Z^n_{\leq t+1}) \mid Z^n_{\leq t} = z^n_{\leq t}}(p)\mu_t^1(dp)+(1-\lambda_t)\int F^{-1}_{-\cvfun^n_{t+1}(Z^n_{\leq t+1}) \mid Z^n_{\leq t} = z^n_{\leq t}}(p)\mu_t^2(dp),
\end{align}
and let 
\begin{align}\label{eq:cvfun_0a}
\cvfun^n_0 =  \lambda_0\int F^{-1}_{\cvfun^n_{1}(Z^n_{1})}(p)\mu_0^1(dp)+(1-\lambda_0)\int F^{-1}_{-\cvfun^n_{1}(Z^n_{1})}(p)\mu_0^2(dp),
\end{align}
where $\mu_t^1, \mu_t^2 \in \calP([0,1])'$. We define $\cvfun_t:\R^t\to \R$ and $\cvfun_0$ analogously without the superscript $n$. 
Note that 
\begin{align*}
\cvfun^n_t(Z^n_{\leq t})&=\sum_{s=1}^{t}X^n_s+V^n_t(X^n)=\phi^n_t(\cvfun^n_{t+1}(Z^n_{\leq t+1})), \; t\geq 1, 
\quad \cvfun^n_0=V^n_0(X)=\phi^n_0(\cvfun^n_1(Z^n_{\leq 1})), \\
\cvfun_t(Z_{\leq t})&=\sum_{s=1}^{t}X_s+V_t(X)=\phi_t(\cvfun_{t+1}(Z_{\leq t+1})), \; t\geq 1, 
\quad \cvfun_0=V_0(X)=\phi_0(\cvfun_1(Z_{\leq 1})). 
\end{align*}
Therefore, the proof is complete once we show the convergence $\lim_{n\to\infty}\cvfun^n_0=\cvfun_0$. 

The argument of the proof is backwards induction. We will show that for $t = T, \dots, 1$, 
\begin{itemize}
\item[{\bf UI($t$):}] $(\calL(\cvfun^n_t(Z^n_{\leq t}) \mid Z^n_{\leq s} = z^n_{\leq s}))_n$ is uniformly integrable for each $s = 1, \dots, t-1$ and $z^n_{\leq s} \to z_{\leq s}$,
\item[{\bf CC($t$):}] $(\cvfun^n_t)_n$ is continuously convergent, i.e. $\cvfun^n_t(z^n_{\leq t}) \to \cvfun_t(z_{\leq t})$, whenever $z^n_{\leq t} \to z_{\leq t}$.
\end{itemize}

\noindent
{\bf Induction base:} {\bf UI($T$)} and {\bf CC($T$)} hold.
 
For $s \in \{1, \dots, T-1\}$ and $z^n_{\leq s} \to z_{\leq s}$, {\bf UI($T$)} follows from \eqref{eq:condxnui} since norms on Euclidean spaces are equivalent and $|\cvfun^n_T(z^n_{\leq T})| = |\sum_{s=1}^T x^n_s| \leq \sum_{s=1}^T |x^n_s|$.   
Moreover, 
$$
\cvfun^n_T(z^n_{\leq T}) = \sum_{s=1}^T x^n_s \to \sum_{s=1}^T x_s = \cvfun_T(z_{\leq T})
$$
whenever $z^n_{\leq T} \to z_{\leq T}$, yielding {\bf CC($T$)}.

\noindent
{\bf Induction step:} {\bf UI($t+1$)}, {\bf CC($t+1$)} together imply {\bf UI($t$)}, {\bf CC($t$)}.

Fix $s\in \{1,\dots,t-1\}$ and sequence $(z^n_{\leq s})_n$ with $z^n_{\leq s}\to z_{\leq s}$. 
Let $\kappa^n_{s+1,t+1}(z^n_{\leq s},\cdot)$ and $\kappa_{s+1,t+1}(z_{\leq s},\cdot)$ be regular versions of the conditional distributions 
$$
\P((Z^n_{s+1},\dots,Z^n_{t+1})^{\trans} \in \cdot \mid Z^n_{\leq s} = z^n_{\leq s})
\quad\text{and}\quad
%$$ 
%and 
%$$
\P((Z_{s+1},\dots,Z_{t+1})^{\trans} \in \cdot \mid Z_{\leq s} = z_{\leq s}).
$$  
By assumption \eqref{eq:contconv}, $\kappa^n_{s+1,t+1}(z^n_{\leq s},\cdot)\weakly \kappa_{s+1,t+1}(z_{\leq s},\cdot)$. 
Since $(\cvfun^n_{t+1}(z^n_{\leq s}, \cdot ))$ is continuously convergent by the induction assumption {\bf CC($t+1$)}, by the generalized continuous mapping theorem (\cite[Theorem 4.27, Exercise 27]{Kallenberg-02}), 
$$
\kappa^n_{s+1,t+1}(z^n_{\leq s}, \cdot) \circ (\cvfun^n_{t+1}(z^n_{\leq s}, \cdot))^{-1} \weakly \kappa_{s+1,t+1}(z_{\leq s}, \cdot) \circ (\cvfun_{t+1}(z_{\leq s}, \cdot))^{-1}, 
$$
i.e.
\begin{align}\label{conv:Weak123}
\calL(\cvfun^n_{t+1}(Z_{\leq t+1}^n) \mid Z^n_{\leq s} = z^n_{\leq s}) \weakly \calL(\cvfun_{t+1}(Z_{\leq t+1}) \mid Z_{\leq s} = z_{\leq s}).
\end{align}
Lemma \ref{lem:gaussian_value} shows that the limiting distribution in \eqref{conv:Weak123} is Gaussian. 
By Lemma \ref{lem:BoundRec}, there exists a constant $c > 0$, such that
$$
|\cvfun^n_t(z)| \leq c \E[|\cvfun^n_{t+1}(Z^n_{\leq t+1})| \mid Z^n_{\leq t} = z].
$$
Hence, it is sufficient to show uniform integrability of
\begin{align}\label{eq:UI23}
\Big(\calL\big(\E[|\cvfun^n_{t+1}(Z^{n}_{\leq t+1})| \mid Z^n_{\leq t}] \mid Z^n_{\leq s}=z^n_{\leq s}\big)\Big)_n.
\end{align}
Setting $\widetilde{W}^{n}=\E[|\cvfun^{n}_{t+1}(Z^n_{\leq t+1})| \mid Z^n_{\leq t}]$ and $W^{n}=|\cvfun^{n}_{t+1}(Z^n_{\leq t+1})|$, this follows from Lemma \ref{lem:TowerUI}, using {\bf UI($t+1$)} and {\bf CC($t+1$)}. Hence, {\bf UI($t$)} holds.

We proceed by showing {\bf CC($t$)}, i.e.~the convergence of 
\begin{align}
\cvfun^n_t(z^n_{\leq t}) &=
\lambda_t\int F^{-1}_{\cvfun^n_{t+1}(Z^n_{\leq t+1}) \mid Z^n_{\leq t} = z_{\leq t}^n }(p) \mu_t^1(dp) \label{eq:RecConv2a} \\
& \quad + (1-\lambda_t)\int F^{-1}_{-\cvfun^n_{t+1}(Z^n_{\leq t+1}) \mid Z^n_{\leq t} = z^n_{\leq t} }(p) \mu_t^2(dp). \label{eq:RecConv2b}
\end{align}
We will prove the convergence assuming that $\mu^1_t$ has a bounded density and that $\supp(\mu^2_t)\subset [a,b]\subset (0,1)$. Other possibilities for $\mu^1_t,\mu^2_t\in \calP([0,1])'$ are handled by the same arguments as those shown below. 

We start by the integral in \eqref{eq:RecConv2a} and will show convergence by an application of Pratt's Lemma \cite{Pratt-60}. 
Let $w : (0, 1) \to \R$, bounded by some $c' > 0$, be the density of $\mu_t^1$, i.e. $\mu_t^1(dp) = w(p)dp$. Let
$$
f_n(p) =  F^{-1}_{\cvfun^n_{t+1}(Z^n_{\leq t+1}) \mid Z^n_{\leq t}= z^n_{\leq t}}(p) w(p).
$$
From \eqref{conv:Weak123} follows convergence 
\begin{align*}
f_n(p) \to F^{-1}_{\cvfun_{t+1}(Z_{\leq t+1}) \mid Z_{\leq t}= z_{\leq t}}(p)w(p) 
\end{align*}
for almost every $p \in (0,1)$. 
By combining \eqref{conv:Weak123} and the continuous mapping theorem applied to the absolute value function follows convergence of upper and lower bounds $l_n(p) \leq f_n(p) \leq u_n(p)$ for almost every $p\in (0,1)$:  
\begin{align*}
l_n(p) &= c'F^{-1}_{-|\cvfun^n_{t+1}(Z^n_{\leq t+1})| \mid Z^n_{\leq t} = z^n_{\leq t} }(p) \to  c'F^{-1}_{-|\cvfun_{t+1}(Z_{\leq t+1})| \mid Z_{\leq t} = z_{\leq t} }(p) = l(p) \\ 
u_n(p) &= c'F^{-1}_{|\cvfun^n_{t+1}(Z^n_{\leq t+1})| \mid Z^n_{\leq t} = z^n_{\leq t} }(p) \to  c'F^{-1}_{|\cvfun_{t+1}(Z_{\leq t+1})| \mid Z_{\leq t} = z_{\leq t} }(p) = u(p).
\end{align*} 
The induction assumption {\bf UI($t+1$)} together with   
\begin{align*}
\calL(|\cvfun^n_{t+1}(Z^n_{\leq t+1})| \mid Z^n_{\leq t} = z^n_{\leq t})
\weakly \calL(|\cvfun_{t+1}(Z_{\leq t+1})| \mid Z_{\leq t} = z_{\leq t}), 
\end{align*}
allow us to conclude from \cite[Theorem 4.11]{Kallenberg-02} that 
\begin{align*}
\int l_n(p) dp = &-c'\E[|\cvfun^n_{t+1}(Z_{\leq t+1}^n)| \mid Z_{\leq t}^n = z_{\leq t}^n] \\ \to 
&-c'\E[|\cvfun_{t+1}(Z_{\leq t+1})| \mid Z_{\leq t} = z_{\leq t}] = \int l(p) dp 
\end{align*}
and 
\begin{align*}
\int u_n(p) dp = \; &c'\E[|\cvfun^n_{t+1}(Z_{\leq t+1}^n)| \mid Z_{\leq t}^n = z_{\leq t}^n] \\ \to \;
&c'\E[|\cvfun_{t+1}(Z_{\leq t+1})| \mid Z_{\leq t} = z_{\leq t}] = \int u(p) dp.
\end{align*}
Hence, by Pratt's Lemma \cite[Theorem 1]{Pratt-60} we have convergence of the integral in \eqref{eq:RecConv2a}, i.e.
\begin{align*}
\int F^{-1}_{\cvfun^n_{t+1}(Z^n_{\leq t+1}) \mid Z^n_{\leq t} = z_{\leq t}^n }(p) \mu_1(dp) 
&= \int f_n(p)dp \to \int f(p) dp \\ 
&= \int F^{-1}_{\cvfun_{t+1}(Z_{\leq t+1}) \mid Z_{\leq t} = z_{\leq t} }(p) \mu_1(dp).
\end{align*}
We now consider the integral in \eqref{eq:RecConv2b}. By Lemma \ref{lem:gaussian_value}, the mapping $p \mapsto F^{-1}_{-\cvfun_{t+1}(Z_{\leq t+1}) \mid Z_{\leq t} = z_{\leq t} }(p)$ is continuous for all $p \in (0,1)$. Therefore, from \eqref{conv:Weak123} follows pointwise convergence 
$$
F^{-1}_{-\cvfun^n_{t+1}(Z^n_{\leq t+1}) \mid Z^n_{\leq t} = z_{\leq t}^n }(p) \to F^{-1}_{-\cvfun_{t+1}(Z_{\leq t+1}) \mid Z_{\leq t} = z_{\leq t} }(p)
$$
for all $p \in (0,1)$. In particular, the convergence holds at the points $a, b \in (0,1)$. Therefore, using the fact that $p \mapsto F^{-1}_{-\cvfun^n_{t+1}(Z^n_{\leq t+1}) \mid Z^n_{\leq t} = z^n_{\leq t} }(p)$ and $p \mapsto F^{-1}_{-\cvfun_{t+1}(Z_{\leq t+1}) \mid Z_{\leq t} = z_{\leq t} }(p)$ are increasing, we can find $d > 0$, such that 
$$
c'' = \max \bigg(\Big|F^{-1}_{-\cvfun_{t+1}(Z_{\leq t+1}) \mid Z_{\leq t} = z_{\leq t} }(a)\Big|, \Big|F^{-1}_{-\cvfun_{t+1}(Z_{\leq t+1}) \mid Z_{\leq t} = z_{\leq t} }(b)\Big| \bigg) + d
$$
uniformly bounds the integrand, i.e.
$$
\Big|F^{-1}_{-\cvfun^n_{t+1}(Z^n_{\leq t+1}) \mid Z^n_{\leq t} = z^n_{\leq t} }(p)\Big| \leq c''
$$
for all $p \in [a, b]$ and for all $n \in \N$. Hence, the bounded convergence theorem yields 
\begin{align*}
\int F^{-1}_{-\cvfun^n_{t+1}(Z^n_{\leq t+1}) \mid Z^n_{\leq t} = z_{\leq t}^n }(p) \mu^2_t(dp) 
\to \int F^{-1}_{-\cvfun_{t+1}(Z_{\leq t+1}) \mid Z_{\leq t} = z_{\leq t} }(p) \mu^2_t(dp).
\end{align*}
Hence, we have shown {\bf CC($t$)} and the proof of the induction step is complete. 

It remains to show convergence of $\cvfun^n_0$. By the continuous convergence of $(\cvfun^n_1)$ and the generalized continuous mapping theorem (\cite[Theorem 4.27]{Kallenberg-02}), it follows that $\calL(\cvfun^n_1(Z_1^n)) \weakly \calL(\cvfun_1(Z_1))$, where the limit is Gaussian (cf. Lemma \ref{lem:gaussian_value}). Therefore the convergence of 
\begin{align*}
\cvfun^n_0=
\lambda_0\int F^{-1}_{\cvfun^n_{1}(Z^n_{1})}(p) \mu_0^1(dp)
 + (1-\lambda_0)\int F^{-1}_{-\cvfun^n_{1}(Z^n_{1}) }(p) \mu_0^2(dp)
\end{align*}
follows from arguments completely analogous to those verifying the induction step. Lemma \ref{lem:gaussian_value} together with the variance decomposition 
\begin{align*}
\Var\bigg(\E\bigg[\sum_{v=u}^{T}X_v\mid Z_{\leq u}\bigg]\mid Z_{\leq u-1}\bigg)
&=\Var\bigg(\sum_{v=u}^{T}X_v\mid Z_{\leq u-1}\bigg)\\
&\quad-\E\bigg[\Var\bigg(\sum_{v=u}^{T}X_v\mid Z_{\leq u}\bigg)\mid Z_{\leq u-1}\bigg]\\
&=\Var\bigg(\sum_{v=u}^{T}X_v\mid Z_{\leq u-1}\bigg)-\Var\bigg(\sum_{v=u}^{T}X_v\mid Z_{\leq u}\bigg)
\end{align*}
completes the proof.
\end{proof}

\begin{proof}[Proof of Theorem \ref{thm:convergentvaluations} assuming \eqref{eq:phiPULLgen} and \eqref{eq:phiPULLgen_n}] 
 The proofs is similar to the above proof of Theorem \ref{thm:convergentvaluations} assuming \eqref{eq:phialtmapping} and \eqref{eq:phinaltmapping}. 
 
Fix $\mu_0, \dots, \mu_{T-1} \in \calP([0,1])'$.  
Let $\crfun^n_T(z^n_{\leq T})=\sum_{s=1}^T x^n_s$ and $\cvfun^n_T(z^n_{\leq T})=\sum_{s=1}^T x^n_s$.  
For $t=1,\dots,T-1$, define mappings $\crfun^n_t,\cvfun^n_t:\supp(\calL(Z^n_{\leq t}))\to \R$ by 
\begin{align}
\crfun^n_{t}(z^n_{\leq t})&=\int F^{-1}_{-\cvfun^n_{t+1}(Z^n_{\leq t+1})\mid Z^n_{\leq t}=z^n_{\leq t}}(p)\mu_t(dp),  \label{eq:ConvRisk3}  \\
\cvfun^n_{t}(z^n_{\leq t})&=\crfun^n_{t}(z^n_{\leq t})-\gamma_t\E\Big[\Big(\big(\crfun^n_{t}(z^n_{\leq t})-\cvfun^n_{t+1}(Z^n_{\leq t+1})\big)^+\Big)^{\beta_t}\mid Z^n_{\leq t}=z^n_{\leq t}\Big]^{1/\beta_t}. \label{eq:ConvUtility3}
\end{align}
Let 
\begin{align}\label{eq:ConvRU3}
\crfun^n_{0}=\int F^{-1}_{-\cvfun^n_{1}(Z^n_{1})}(p)\mu_0(dp),\quad 
\cvfun^n_{0}=\crfun^n_{0}-\gamma_0\E\Big[\Big(\big(\crfun^n_{0}-\cvfun^n_{1}(Z^n_{1})\big)^+\Big)^{\beta_0}\Big]^{1/\beta_0}.
\end{align}
Define $\crfun_t,\cvfun_t:\R^t\to\R$ and $\crfun_0,\cvfun_0$ analogously without the superscript $n$. 
The proof is complete once we show the convergence $\lim_{n\to\infty}\cvfun^n_0=\cvfun_0$. 

The argument of the proof is backwards induction. We will show that for $t = T, \dots, 1$, 
\begin{itemize}
\item[{\bf UI($t$):}] $(\calL(\cvfun^n_t(Z^n_{\leq t}) \mid Z^n_{\leq s} = x^n_{\leq s}))_n$ is uniformly integrable for each $s = 1, \dots, t-1$ and $z^n_{\leq s} \to z_{\leq s}$,
\item[{\bf CC($t$):}] $(\cvfun^n_t)_n$ is continuously convergent, i.e. $\cvfun^n_t(z^n_{\leq t}) \to \cvfun_t(z_{\leq t})$, whenever $z^n_{\leq t} \to z_{\leq t}$.
\end{itemize}

\noindent
{\bf Induction base:} {\bf UI($T$)} and {\bf CC($T$)} hold. The argument verifying the induction base is identical to that in the proofs of Theorem \ref{thm:convergentvaluations} assuming \eqref{eq:phialtmapping} and \eqref{eq:phinaltmapping}. 

\noindent
{\bf Induction step:} {\bf UI($t+1$)}, {\bf CC($t+1$)} together imply {\bf UI($t$)}, {\bf CC($t$)}.
The argument verifying that 
\begin{align}\label{conv:Weak123456}
\calL(\cvfun^n_{t+1}(Z^n_{\leq t+1}) \mid Z^n_{\leq s} = z^n_{\leq s}) \weakly \calL(\cvfun_{t+1}(Z_{\leq t+1}) \mid Z_{\leq s} = z_{\leq s}).
\end{align}
holds is identical to that in the proofs of Theorem \ref{thm:convergentvaluations} assuming \eqref{eq:phialtmapping} and \eqref{eq:phinaltmapping}. 
By Lemma \ref{lem:BoundUIUtility}, there exists a constant $c > 0$, such that
$$
|\cvfun^n_{t}(z)| \leq c \E[|\cvfun^n_{t+1}(Z^n_{\leq t+1})| \mid Z^n_{\leq t} = z].
$$
Hence, to verify {\bf UI($t$)} it is sufficient to show uniform integrability of 
$$
\Big(\calL\Big(\E[|\cvfun^n_{t+1}(Z^n_{\leq t+1})| \mid Z^n_{\leq t}] \mid Z^n_{\leq s} = z^n_{\leq s} \Big)\Big)_{n\in\N}
$$
and after setting $W^n = |\cvfun^n_{t+1}(Z^n_{\leq t+1})|$, $\widetilde{W}^n = \E[|\cvfun^n_{t+1}(Z_{\leq t+1}^n)| \mid Z^n_{\leq t}]$,
 {\bf UI($t$)} follows from {\bf UI($t+1$)} and \eqref{conv:Weak123456} by applying Lemma \ref{lem:TowerUI}.

We proceed by showing {\bf CC($t$)}, i.e.~the convergence of \eqref{eq:ConvUtility3}. 
We first consider \eqref{eq:ConvRisk3}. 
If $\mu_t$ admits a bounded density, convergence $\rho_t^n(z^n_{\leq t}) \to \rho_t(z_{\leq t})$ can be shown by Pratt's Lemma, together with \eqref{conv:Weak123456} and {\bf UI($t+1$)} (cf. Proof of Theorem \ref{thm:convergentvaluations} assuming \eqref{eq:phialtmapping} and \eqref{eq:phinaltmapping}). If $\mu_t$ satisfies $\supp(\mu_t) \subset [a, b]$ with $0<a<b<1$, we first note that we have pointwise convergence $F^{-1}_{-\cvfun^n_{t+1}(Z^n_{\leq t})\mid Z^n_{\leq t} = z^n_{\leq t}}(p) \to F^{-1}_{-\cvfun_{t+1}(Z_{\leq t})\mid Z_{\leq t} = z_{\leq t} }(p) $ for all $p \in (0,1)$ since the limiting distribution in \eqref{conv:Weak123456} is Gaussian (cf. Lemma \ref{lem:gaussian_value}). There exists $c > 0$ such that
\begin{align*}
\sup_n\sup_{p \in [a,b]}|F^{-1}_{-\cvfun^n_{t+1}(Z^n_{\leq t+1}) \mid Z^n_{\leq t} = z^n_{\leq t}}(p)| 
 \leq \sup_n c \E[|\cvfun^n_{t+1}(Z^n_{\leq t+1}) | \mid Z^n_{\leq t} = z^n_{\leq t}]
\end{align*}
and the right-hand side is finite. 
The inequality is due to properties of quantile functions (cf. Proof of Lemma \ref{lem:BoundRec}). The finiteness is because of the convergence of the conditional expectation which is due to \eqref{conv:Weak123456} and {\bf UI($t+1$)}. Hence, dominated convergence yields $\crfun^n_t(z^n_{\leq t}) \to \crfun_t(z_{\leq t})$.

We now show convergence of the conditional expectation in \eqref{eq:ConvUtility3} which can be written in terms of 
\begin{align*}%\label{eq:ExpectationUtility1}
\E \Big[h_n\big(\cvfun^n_{t+1}(Z^n_{\leq t+1})\big) \mid Z^n_{\leq t} = z^n_{\leq t}\Big],
\end{align*}
where $h_n(w) = \big(\big(\crfun^n_t(z^n_{\leq t}) - w\big)^{+}\big)^{\beta_t}$ converges continuously towards $h(w) = \big(\big(\crfun_t(z_{\leq t}) - w\big)^{+}\big)^{\beta_t}$. Moreover, since 
$$
|h_n(w)| \leq \bigg(\sup_n(\crfun^n_t(z^n_{\leq t})) + |w|\bigg)^{\beta_t}, \quad \sup_n(\crfun^n_t(z^n_{\leq t})) < \infty,
$$
and $\beta_t\in (0,1]$, uniform integrability of $\calL\big(h_n\big(\cvfun^n_{t+1}(Z^n_{\leq t+1})\big) \mid Z^n_{\leq t} = z^n_{\leq t}\big)$ follows from {\bf UI($t+1$)}. This means convergence 
\begin{align*}%\label{eq:ExpectationUtility2}
\E\Big[h_n\big(\cvfun^n_{t+1}(Z^n_{\leq t+1})\big) \mid Z^n_{\leq t} = z^n_{\leq t}\Big]
\to \E\Big[h\big(\cvfun_{t+1}(Z_{\leq t+1})\big) \mid Z_{\leq t} = z_{\leq t}\Big]
\end{align*}
which implies {\bf CC($t$)} and the proof of the induction step is complete. 

Finally, the proof of the convergence $\lim_{n\to\infty}\cvfun^n_0=\cvfun_0$ follows by arguments analogous to those in the proof of Theorem \ref{thm:convergentvaluations} assuming \eqref{eq:phialtmapping} and \eqref{eq:phinaltmapping}. 
\end{proof}

\begin{lemma}\label{lem:quant_int_bound}
Let $W$ be a random variable with quantile function $F^{-1}(p)=\min\{x\in\R:F_{W}(x)\geq p\}$. If $\mu\in\calP([0,1])'$, then there exists $c\in (0,\infty)$ such that 
\begin{align*}
\Big |\int F^{-1}_{W}(p)\mu(dp) \Big | \leq c\E[|W|].
\end{align*}
\end{lemma}

\begin{proof}
We have the general bounds $l,u$ given by 
\begin{align*}
l=\int F^{-1}_{-|W|}(p)\mu(dp) \leq \int F^{-1}_{W}(p)\mu(dp) \leq \int F^{-1}_{|W|}(p)\mu(dp)=u.
\end{align*}
Moreover, 
\begin{align*}
F^{-1}_{-|W|}(p)=-F^{-1}_{|W|}((1-p)+)=-\lim_{v\downarrow 1-p}F^{-1}_{|W|}(v).
\end{align*}
For any $p\in (0,1)$, 
\begin{align*}
pF^{-1}_{|W|}(1-p)&\leq \int_{1-p}^{1}F^{-1}_{|W|}(v)dv\leq \E[|W|], \\
pF^{-1}_{-|W|}(p)&\geq \int_{0}^{p}F^{-1}_{-|W|}(v)dv\geq -\E[|W|].
\end{align*}
Consequently, if $\supp(\mu)\subset [a,b]\subset (0,1)$, then 
\begin{align*}
l\geq -\frac{1}{a}\E[|W|], \quad u\leq \frac{1}{1-b}\E[|W|], \quad \Big | \int F^{-1}_{W}(p)\mu(dp)\Big | \leq \max\Big(\frac{1}{a},\frac{1}{1-b}\Big)\E[|W|].
\end{align*}
If $\mu$ has a density bounded by $c'\in (0,\infty)$, then 
\begin{align*}
l\geq -c'\E[|W|], \quad u\leq c'\E[|W|], \quad \Big | \int F^{-1}_{W}(p)\mu(dp)\Big | \leq c'\E[|W|].
\end{align*}
\end{proof}

\begin{lemma}\label{lem:BoundRec}
If $\cvfun^n_t$ is given by \eqref{eq:cvfun_ta} and \eqref{eq:cvfun_0a}, then exists $c\in (0,\infty)$ such that 
\begin{align*} 
|\cvfun^n_{0}| \leq c \E[|\cvfun_{1}^n(Z^n_{1})|] \quad \text{and} \quad 
|\cvfun^n_{t}(z)| \leq c \E[|\cvfun_{t+1}^n(Z^n_{\leq t+1})| \mid Z^n_{\leq t} = z]
\end{align*}
for any $z\in\supp(\calL(Z^{n}_{\leq t}))$. 
\end{lemma}

\begin{proof}
We proof the inequality for $t\geq 1$. The argument for $t=0$ is identical. 
Since 
\begin{align*} 
|\cvfun^n_t(z)| 
\leq \Big |\int F^{-1}_{\cvfun^n_{t+1}(Z^n_{\leq t+1}) \mid Z^n_{\leq t} = z}(p) \mu_t^1(dp) \Big | 
+ \Big | \int F^{-1}_{-\cvfun^n_{t+1}(Z^n_{\leq t+1}) \mid Z^n_{\leq t} = z}(p) \mu_t^2(dp) \Big |, 
\end{align*}
we can apply Lemma \ref{lem:quant_int_bound}, bounding the right-hand side by $c\E[|\cvfun^n_{t+1}(Z^n_{\leq t})| \mid Z^n_{\leq t} = z]$ for some $c$. 
\end{proof}

The following result is similar to Proposition 6 in \cite{engsner2017}.

\begin{lemma}\label{lem:gaussian_value}
Let $(Z_t)_{t=1}^T=(X_t,Y_t)_{t=1}^T$ be a Gaussian process, let $\calF_0=\{\Omega,\emptyset\}$, and let $\calF_t=\sigma(Z_{\leq t})$ for $t\geq 1$.   
Let $V_T=0$ and, for $t\in\{0,1,\dots,T-1\}$, let $V_t=\phi_t(X_{t+1}+V_{t+1})$, where $\phi_t$ is positive homogeneous and conditionally cash additive. 
Suppose that for each $t$,  
\begin{align*}
\text{if}\quad \P(W \in \cdot \mid \calF_t)=\P(\widetilde{W}\in \cdot \mid \calF_t), \quad\text{then}\quad 
\phi_t(W)=\phi_t(\widetilde{W}),
\end{align*}
and if $W$ is independent of $\calF_t$, then $\phi_t(W)$ is nonrandom. 
Then 
\begin{align*}
V_t=\E\bigg[\sum_{u=t+1}^T X_{u}\mid \calF_t\bigg]+\sum_{u=t+1}^T \phi_{u-1}(\varepsilon_u)\Var\bigg(\E\bigg[\sum_{v=u}^T X_v\mid \calF_u\bigg]\mid\calF_{u-1}\bigg)^{1/2},
\end{align*}
where $\varepsilon_u$ is standard normal and independent of $\calF_{u-1}$. 
\end{lemma}

\begin{proof}
The lemma is proved by backward induction. We prove the induction step. Suppose that $X_{t+1}+V_{t+1}$ and $Z_{\leq t}$ are jointly Gaussian. Then there exists a standard normal $\varepsilon_{t+1}$ independent of $\calF_t$ such that 
\begin{align*}
\P(X_{t+1}+V_{t+1}\in \cdot \mid \calF_t)=\P\Big(\E[X_{t+1}+V_{t+1}\mid \calF_t]+\varepsilon_{t+1}\Var(X_{t+1}+V_{t+1}\mid \calF_t)^{1/2}\in \cdot \mid\calF_t\Big)
\end{align*}
and it follows from properties of the multivariate normal distribution that the conditional variance $\Var(X_{t+1}+V_{t+1}\mid \calF_t)$ is nonrandom. Hence, by positive homogeneity and conditional cash additivity, 
\begin{align*}
\phi_t(X_{t+1}+V_{t+1})=\E[X_{t+1}+V_{t+1}\mid \calF_t]+\phi_t(\varepsilon_{t+1})\Var(X_{t+1}+V_{t+1}\mid \calF_t)^{1/2}.
\end{align*}
By the induction assumption (in the nontrivial case $t+1<T$),   
\begin{align*}
V_{t+1}=\E\bigg[\sum_{u=t+2}^T X_{u}\mid \calF_{t+1}\bigg]+\sum_{u=t+2}^T \phi_{u-1}(\varepsilon_u)\Var\bigg(\E\bigg[\sum_{v=u}^T X_v\mid \calF_u\bigg]\mid\calF_{u-1}\bigg)^{1/2}. 
\end{align*}
Summing up the terms and using the tower property of conditional expectations complete the proof. 
\end{proof}

\begin{lemma}\label{lem:TowerUI}
Let $s \in \{1, \dots, t-1\}$ and let $(W^n)_{n\in\N}$ be a sequence of non-negative random variables such that 
\begin{align}\label{eq:wczTUI}
\calL(W^n \mid Z^n_{\leq s} = z^n_{\leq s}) \weakly \calL(W \mid Z_{\leq s} = z_{\leq s}) \quad\text{as } n\to\infty
\end{align}
and 
\begin{align}\label{eq:uizTUI}
\Big(\calL(W^n \mid Z^n_{\leq s} = z^n_{\leq s})\Big)_{n\in\N} \quad\text{is uniformly integrable.}
\end{align}
If $\widetilde{W}^n = \E[W^n \mid Z^n_{\leq t}]$, then $\big(\calL(\widetilde{W}^n \mid Z^n_{\leq s} = z^n_{\leq s}) \big)_{n\in\N}$ is uniformly integrable.
\end{lemma}

\begin{proof}
First note that
\begin{align*}
\E\Bigl[\widetilde{W}^nI\{\widetilde{W}^n \geq r\}\mid Z^n_{\leq s} \Bigr] 
& =\E\Bigl[\E[W^n \mid Z^n_{\leq t}]I\{\widetilde{W}^n \geq r\}\mid Z^n_{\leq s} \Bigr] \\
& =\E\Bigl[\E[W^n I\{\widetilde{W}^n \geq r\} \mid Z^n_{\leq t}]\mid Z^n_{\leq s}\Bigr]\\
& =\E\Bigl[W^n I\{\widetilde{W}^n \geq r\} \mid Z^n_{\leq s}\Bigr],
\end{align*}
i.e. $\E[\widetilde{W}^{n}I\{\widetilde{W}^{n} \geq r\}\mid Z^n_{\leq s}=z^n_{\leq s}]=\E[W^{n} I\{\widetilde{W}^{n} \geq r\} \mid Z^n_{\leq s}=z^n_{\leq s}]$. 
%\begin{align}\label{eq:Tower}
%\E[\widetilde{W}^{n}I\{\widetilde{W}^{n} \geq r\}\mid Z^n_{\leq s}=z^n_{\leq s}]
%&=\E[W^{n} I\{\widetilde{W}^{n} \geq r\} \mid Z^n_{\leq s}=z^n_{\leq s}].
%\end{align}
For each $\gamma_r > 0$, 
\begin{align}\label{bound:23}
&\E[W^{n}I\{\widetilde{W}^{n} \geq r\}\mid Z^n_{\leq s}=z^n_{\leq s}]  \nonumber \\
&\quad= \E[W^{n}I\{\widetilde{W}^{n} \geq r\} I\{W^n \leq \gamma_r\}\mid Z^n_{\leq s}=z^n_{\leq s}] \nonumber \\
&\quad\quad + \E[W^{n}I\{\widetilde{W}^{n} \geq r\} I\{W^n > \gamma_r\}\mid Z^n_{\leq s}=z^n_{\leq s}] \nonumber \\
&\quad \leq \gamma_{r}\P(\widetilde{W}^{n} \geq r \mid Z^n_{\leq s}=z^n_{\leq s}) +\E[W^{n}I\{W^{n}>\gamma_{r}\}\mid Z^n_{\leq s}=z^n_{\leq s}].
\end{align}
By \eqref{eq:wczTUI} and \eqref{eq:uizTUI}, 
$$
\E[\widetilde{W}^{n}I\{\widetilde{W}^{n} \geq r\}\mid Z^n_{\leq s}=z^n_{\leq s}] < \infty \quad\text{for all } n\in\N,
$$
and similarly with $W^{n}$ instead of $\widetilde{W}^{n}$. 
Therefore, uniform integrability of $\big(\calL(\widetilde{W}^n \mid Z^n_{\leq s} = z^n_{\leq s}) \big)_{n\in\N}$ is equivalent to (cf.~\cite{Kallenberg-02} p.~67) 
\begin{align*}
\lim_{r\to\infty}\limsup_{n\to\infty}\E[\widetilde{W}^{n}I\{\widetilde{W}^n\geq r\}\mid Z^n_{\leq s}=z^n_{\leq s}]=0 
\end{align*}
which follows if we show that there exists $\gamma_{r}$ such that $\lim_{r \to \infty}\limsup_{n\to\infty}[\eqref{bound:23}]=0$.  
By the uniform integrability assumption \eqref{eq:uizTUI} and the finiteness of all the terms, 
\begin{align*}
\lim_{r\to\infty}\limsup_{n\to\infty}\E[W^{n}I\{W^{n}>\gamma_{r}\}\mid Z^n_{\leq s}=z^n_{\leq s}]=0. 
\end{align*}
Therefore, it only remains to show that 
$$
\lim_{r \to \infty}\limsup_{n\to\infty}\gamma_{r}\P(\widetilde{W}^n \geq r \mid Z^n_{\leq s} = z^n_{\leq s})=0
$$ 
for some sequence $\gamma_{r}\to\infty$. By Markov's inequality,
\begin{align*}
\P(\widetilde{W}^n \geq r \mid Z^n_{\leq s} = z^n_{\leq s})
&\leq \frac{ \E[\widetilde{W}^n \mid Z^n_{\leq s} = z^n_{\leq s}]}{r} \\
&= \frac{\E[W^n \mid Z^n_{\leq s} = z^n_{\leq s}]}{r} \\
&\to \frac{\E[W \mid Z_{\leq s} = z_{\leq s}]}{r}, 
\end{align*}
where the convergence in the last step follows from Lemma 4.11 in \cite{Kallenberg-02}, using \eqref{eq:wczTUI} and \eqref{eq:uizTUI}.   Letting e.g.~$\gamma_r = \sqrt{r}$, the first term in \eqref{bound:23} can therefore be made arbitrarily small by choosing $r$ sufficiently large.
\end{proof}

\begin{lemma}\label{lem:BoundUIUtility}
If $\crfun^n_t,\cvfun^n_t$ are given by \eqref{eq:ConvRisk3}, \eqref{eq:ConvUtility3} and \eqref{eq:ConvRU3}, then exists $c\in (0,\infty)$ such that 
\begin{align*} 
|\cvfun^n_{0}| \leq c \E[|\cvfun_{1}^n(Z^n_{1})|] \quad \text{and} \quad 
|\cvfun^n_{t}(z)| \leq c \E[|\cvfun_{t+1}^n(Z^n_{\leq t+1})| \mid Z^n_{\leq t} = z]
\end{align*}
for any $z\in\supp(\calL(Z^{n}_{\leq t}))$. 
\end{lemma}

\begin{proof}[Proof of Lemma \ref{lem:BoundUIUtility}]
Consider $t\geq 1$. 
Note that, by Lemma \ref{lem:quant_int_bound}, there exists $c' > 0$, such that
$$
|\crfun^n_t(z)| \leq c' \E[|\cvfun^n_{t+1}(Z^n_{\leq t+1})| \mid Z^n_{\leq t} = z].
$$
Hence, 
\begin{align*}
|\cvfun^n_t(z)| &\leq c' \E[ |\cvfun^n_{t+1}(Z^n_{\leq t+1})| \mid Z^n_{\leq t} = z] \\
& \quad +\gamma_t\E \Big[ \Big( \big(c' \E[ |\cvfun^n_{t+1}(Z^n_{\leq t+1})| \mid Z^n_{\leq t} = z]-\cvfun^n_{t+1}(Z^n_{\leq t+1})\big)^+ \Big)^{\beta_t} \mid Z^n_{\leq t} = z\Big]^{1/\beta_t}.
\end{align*}
By Jensen's inequality, the second summand above is bounded by
\begin{align*}
&\gamma_t\E \Big[ \big(c' \E[ |\cvfun^n_{t+1}(Z^n_{\leq t+1})| \mid Z^n_{\leq t} = z]- \cvfun^n_{t+1}(Z^n_{\leq t+1}) \big)^+\mid Z^n_{\leq t} = z  \Big] \\
&\quad \leq \gamma_t\E \Big[ c' \E[ |\cvfun^n_{t+1}(Z^n_{\leq t+1})| \mid Z^n_{\leq t} = z] + |\cvfun^n_{t+1}(Z^n_{\leq t+1})| \mid Z^n_{\leq t} = z\Big] \\
& \quad = (c' + 1) \gamma_t \E[ |\cvfun^n_{t+1}(Z^n_{\leq t+1})| \mid Z^n_{\leq t} = z].
\end{align*}
Hence, the statement holds with $c=c'+\gamma_t(c'+1)$. 
The inequality 
$|\cvfun^n_{0}| \leq c \E[|\cvfun_{1}^n(Z^n_{1})|]$ is shown by the same arguments. 
\end{proof}

\begin{proof}[Proof of Theorem \ref{thm:CondExpVar}]
The proof proceeds in a similar manner as the proof of Theorem \ref{thm:convergentvaluations}. For $t=1,\dots,T$, consider mappings $\cvfun^n_t:\supp(\calL(Z^n_{\leq t}))\to \R$ given by 
\begin{align*}
\cvfun^n_t(z^n_{\leq t}) = 
 \left\{\begin{array}{ll}
\sum_{s=1}^T x^n_s, & t=T,\\
 \E \big[\cvfun^n_{t+1}(Z^n_{\leq t+1}) \mid Z^n_{\leq t} = z^n_{\leq t}\big] + c_t \Var\big(\cvfun^n_{t+1}(Z^n_{\leq t+1}) \mid Z^n_{\leq t} = z^n_{\leq t}\big)^{1/2}, & t<T,
\end{array}\right.
\end{align*}
and 
\begin{align*}
\cvfun^n_0 =  \E \big[\cvfun^n_{1}(Z^n_{1}) \big] + c_0 \Var\big(\cvfun^n_{1}(Z^n_{1})\big)^{1/2}. 
\end{align*}
We define $\cvfun_t:\R^t\to \R$ and $\cvfun_0$ analogously without the superscript $n$.  
The proof is complete once we show the convergence $\lim_{n\to\infty}\cvfun^n_0=\cvfun_0$. 

The argument of the proof is backwards induction. We will show that for $t = T, \dots, 1$, 
\begin{itemize}
\item[{\bf UI($t$):}] $(\calL(\cvfun^n_t(Z^n_{\leq t})^2 \mid Z^n_{\leq s} = z^n_{\leq s}))_n$ is uniformly integrable for each $s = 1, \dots, t-1$ and $z^n_{\leq s} \to z_{\leq s}$,
\item[{\bf CC($t$):}] $(\cvfun^n_t)_n$ is continuously convergent, i.e. $\cvfun^n_t(z^n_{\leq t}) \to \cvfun_t(z_{\leq t})$, whenever $z^n_{\leq t} \to z_{\leq t}$.
\end{itemize}

\noindent
{\bf Induction base:} {\bf UI($T$)} and {\bf CC($T$)} hold.
 
For $s \in \{1, \dots, T-1\}$ and $z^n_{\leq s} \to z_{\leq s}$, {\bf UI($T$)} follows immediately from \eqref{eq:condxnuiII} since $\cvfun^n_{T}(z^n_{\leq T}) = \sum_{s=1}^T x^n_s$.  
Moreover, 
$$
\cvfun^n_{T}(z^n_{\leq T}) = \sum_{s=1}^T x^n_s \to \sum_{s=1}^T x_s = \cvfun_T(z_{\leq T})
$$
whenever $z_{\leq T}^n \to z_{\leq T}$, yielding {\bf CC($T$)}.

\noindent
{\bf Induction step:}{\bf UI($t+1$)}, {\bf CC($t+1$)} together imply {\bf UI($t$)}, {\bf CC($t$)}.
Fix $s\in \{1,\dots,t-1\}$ and a sequence $(z^n_{\leq s})_n$ with $z^n_{\leq s}\to z_{\leq s}$. 
By the same argument as in the proof of Theorem \ref{thm:convergentvaluations}, 
\begin{align}\label{conv:Weak1234}
\calL(\cvfun^n_{t+1}(Z^n_{\leq t+1}) \mid Z^n_{\leq s} = z^n_{\leq s}) \weakly \calL(\cvfun_{t+1}(Z_{\leq t+1}) \mid Z_{\leq s} = z_{\leq s})
\end{align}
which implies 
\begin{align}\label{conv:Weak12345}
\calL((\cvfun^n_{t+1}(Z^n_{\leq t+1}))^2 \mid Z^n_{\leq s} = z^n_{\leq s}) \weakly \calL((\cvfun_{t+1}(Z_{\leq t+1}))^2 \mid Z_{\leq s} = z_{\leq s}).
\end{align}
Lemma \ref{lem:gaussian_value} shows that the limiting distribution in \eqref{conv:Weak1234} is Gaussian. We can bound
\begin{align*}
\cvfun^n_t(z)^2 &\leq 2\E[|\cvfun^n_{t+1}(Z^n_{\leq t+1})| \mid Z^n_{\leq t} = z]^2 + 2c_t^2\E[|\cvfun^n_{t+1}(Z^n_{\leq t+1})|^2 \mid Z^n_{\leq t} = z] \\
&\leq (2+2c_t^2)\E[|\cvfun^n_{t+1}(Z^n_{\leq t+1})|^2 \mid Z^n_{\leq t} = z].
\end{align*}
Therefore, 
\begin{align*}
\E[|\cvfun^n_t(Z^n_{\leq t})|^2 \, \mid Z^n_{\leq s} = z^n_{\leq s}]
&\leq (2+2c_t^2)\E[\E[|\cvfun^n_{t+1}(Z^n_{\leq t+1})|^2 \, \mid Z^n_{\leq t}]\mid Z^n_{\leq s} = z^n_{\leq s}]\\
&=(2+2c_t^2)\E[|\cvfun^n_{t+1}(Z^n_{\leq t+1})|^2 \, \mid Z^n_{\leq s} = z^n_{\leq s}].
\end{align*}
Hence, it is sufficient to show uniform integrability of
\begin{align}\label{eq:UI234}
\Big(\calL\big(\E[|\cvfun^n_{t+1}(Z^n_{\leq t+1})|^2 \mid Z^n_{\leq t}] \mid Z^n_{\leq s}=z^n_{\leq s}\big)\Big)_n.
\end{align}
Setting $\widetilde{W}^{n}=\E[|\cvfun^n_{t+1}(Z^n_{\leq t+1})|^2 \mid Z^n_{\leq t}]$ and $W^{n}=|\cvfun^n_{t+1}(Z^n_{\leq t+1})|^2$, this follows from Lemma \ref{lem:TowerUI} together with \eqref{conv:Weak12345} and {\bf UI($t+1$)}. Hence, {\bf UI($t$)} holds.

We proceed by showing {\bf CC($t$)}. Let $z^n_{\leq t} \to z_{\leq t}$. We need to show convergence of
\begin{align*}
\cvfun^n_t(z^n_{\leq t})  
&= \E[\cvfun^n_{t+1}(Z^n_{\leq t+1}) \mid Z^n_{\leq t} = z^n_{\leq t}] \\  
& \quad + c_t \left(\E[(\cvfun^n_{t+1}(Z^n_{\leq t+1}))^2 \mid Z^n_{\leq t} = z^n_{\leq t}] - \E[\cvfun^n_{t+1}(Z^n_{\leq t+1}) \mid Z^n_{\leq t} = z^n_{\leq t}]^2 \right)^{1/2}.
\end{align*}
Indeed, by \eqref{conv:Weak1234} and \eqref{conv:Weak12345},
\begin{align*}
&\calL(\cvfun^n_{t+1}(Z^n_{\leq t+1}) \mid Z^n_{\leq t}  = z^n_{\leq t}) \weakly \calL(\cvfun_{t+1}(Z_{\leq t+1}) \mid Z_{\leq t} = z_{\leq t}), \\
&\calL((\cvfun^n_{t+1}(Z^n_{\leq t+1}))^2 \mid Z^n_{\leq t}  = z^n_{\leq t}) \weakly \calL((\cvfun_{t+1}(Z_{\leq t+1}))^2 \mid Z_{\leq t} = z_{\leq t}),
\end{align*}
and by \cite[Lemma 4.11]{Kallenberg-02} applied to {\bf UI($t+1$)}, it follows that
$$
\E[\cvfun^n_{t+1}(Z^n_{\leq t+1}) \mid Z^n_{\leq t} = z^n_{\leq t}] \to \E[\cvfun_{t+1}(Z_{\leq t+1}) \mid Z_{\leq t} = z_{\leq t}]
$$
and
$$
\E[(\cvfun^n_{t+1}(Z^n_{\leq t+1}))^2 \mid Z^n_{\leq t} = z^n_{\leq t}] \to \E[(\cvfun_{t+1}(Z_{\leq t+1}))^2 \mid Z_{\leq t} = z_{\leq t}],
$$
yielding {\bf CC($t$)} which completes the proof of the induction step.

As a result of the induction, 
$$
\calL(\cvfun^n_1(Z^n_1)) \weakly \calL(\cvfun_1(Z_1))
$$
and $(\cvfun^n_1(Z^n_1))_n, ((\cvfun^n_1(Z^n_1))^2)_n$ are uniformly integrable. Therefore,
\begin{align*}
\cvfun^n_0 &= \E[\cvfun^n_{1}(Z^n_1)] + c_0\left(\E[(\cvfun^n_{1}(Z^n_1))^2] - \E[\cvfun^n_1(Z^n_1)]^2 \right)^{1/2} \\
 &\to \E[\cvfun_{1}(Z_1)] + c_0\left(\E[(\cvfun_{1}(Z_1))^2] - \E[\cvfun_1(Z_1)]^2 \right)^{1/2} = \cvfun_0.
\end{align*}
\end{proof}

\begin{lemma}\label{lem:V0monotonicity}
Let $c\in \R^T_+$ satisfy $\sum_{t=1}^Tc_t=1$ and $c_{t+1}\leq c_t$ for every $t$. Then $c$ solves 
\begin{align}\label{optproblem}
\max_{d\in \mathcal{D}_{c,T}} \sum_{t=1}^T d_t^{1/2}, 
\quad \mathcal{D}_{c,T}=\bigg\{d\in\R^T_+:\sum_{t=1}^T d_t=1,\sum_{s=1}^t d_s \geq \sum_{s=1}^t c_s \text{ for } t=1,\dots,T\bigg\}.
\end{align}
\end{lemma}

\begin{proof}
Since $\mathcal{D}_{c,T}$ is a convex set and the objective function is concave, we have a convex optimization problem. Consequently, it is sufficient to verify that $c\in\mathcal{D}_{c,T}$ is a local optimum. 
By assumption either $c_1=1$ and $c_t=0$ for $t\geq 2$ or there exists $t_0\in\{2,\dots,T\}$ such that $c_1\geq \dots \geq c_{t_0}>0$ and $c_{t_0+1}=\dots=c_{T}=0$. The first case has the trivial maximizer $d_1=1$ and $d_t=0$ for $t\geq 2$. Hence, it is sufficient to consider only the second case for a fixed $c$ and $t_0\in\{2,\dots,T\}$. In this case, \eqref{optproblem} is equivalent to, with $c=(c_1,\dots,c_{t_0})$,  
\begin{align}\label{modoptproblem}
\max_{d\in \mathcal{D}_c} \sum_{t=1}^{t_0} d_t^{1/2}, 
\quad \mathcal{D}_c=\bigg\{d\in\R^{t_0}_+:\sum_{t=1}^{t_0} d_t=1,\sum_{s=1}^t d_s \geq \sum_{s=1}^t c_s \text{ for } t=1,\dots,t_0\bigg\}.
\end{align}
Since $\mathcal{D}_c$ is a convex set and the objective function is concave, we have a convex optimization problem. Consequently, it is sufficient to verify that $c\in\mathcal{D}_c$ is a local optimum. 
Let 
\begin{align*}
\mathcal{X}=\bigg\{x\in\R^{t_0}:\sum_{t=1}^{t_0} x_t=0,\sum_{s=1}^t x_s \geq 0 \text{ for } t=1,\dots,t_0\bigg\}
\end{align*}
and note that $\mathcal{D}_c\subset c+\mathcal{X}$. 
Consider the set of vectors $\{b_1,\dots,b_{t_0-1}\}\subset \R^{t_0}$ given by 
\begin{align*}
b_{i,i}=1, \; b_{i,i+1}=-1, \; b_{i,j}=0 \text{ for } j\notin \{i,i+1\}.
\end{align*}
We claim that $\mathcal{X}=\textrm{span}_+\{b_1,\dots,b_{t_0-1}\}$, where 
\begin{align*}
\textrm{span}_+\{b_1,\dots,b_{t_0-1}\}=\bigg\{\sum_{k=1}^{t_0-1}\lambda_k b_k : \lambda_1,\dots,\lambda_{t_0-1}\geq 0\bigg\}.
\end{align*}
To show that $\textrm{span}_+\{b_1,\dots,b_{t_0-1}\}\subseteq \mathcal{X}$ it is sufficient to note that 
\begin{align*}
&\sum_{s=1}^{t_0}\sum_{k=1}^{t_0-1}\lambda_k b_{k,s}=\sum_{k=1}^{t_0-1}\lambda_k\sum_{s=1}^{t_0}b_{k,s}=0, \\
&\sum_{s=1}^t\sum_{k=1}^{t_0-1}\lambda_k b_{k,s}=\sum_{k=1}^{t_0-1}\lambda_k\sum_{s=1}^tb_{k,s}=\lambda_t\geq 0.
\end{align*}
We now show that $\mathcal{X}\subseteq\textrm{span}_+\{b_1,\dots,b_{t_0-1}\}$. Take $x\in\mathcal{X}$ and set $\lambda_1=x_1$ (noting that $x_1\geq 0$) and $\lambda_k=\lambda_{k-1}+x_k$ for $k\geq 2$. Hence, $\lambda_k\geq 0$ for every $k$ and $\sum_{k=1}^{t_0-1}\lambda_kb_k=x$. 

Let $g(d)=\sum_{t=1}^{t_0} d_t^{1/2}$ and note that $g$ is well defined and concave on $\mathcal{D}_c$, and has a well defined gradient at $c$: 
\begin{align*}
\nabla g(c)^{\trans}=\frac{1}{2}(c_1^{-1/2},\dots,c_{t_0}^{-1/2}).
\end{align*} 
Take $d\in\mathcal{D}_c$ and notice that $d=c+x$ for $x\in\mathcal{X}$. Since $g$ is concave, 
\begin{align*}
g(d)-g(c)\leq \nabla g(c)^{\trans}x
=\sum_{k=1}^{t_0-1}\lambda_k\nabla g(c)^{\trans}b_k
=\frac{1}{2}\sum_{k=1}^{t_0-1}\lambda_k(c_k^{-1/2}-c_{k+1}^{-1/2})\leq 0,
\end{align*}
where the last inequality holds since, for every $k$, $\lambda_k\geq 0$ and $c_k^{-1/2}-c_{k+1}^{-1/2}\leq 0$. 
Hence, we have shown that $g(d)\leq g(c)$ and the proof is complete. 
\end{proof}

\begin{proof}[Proof of Theorem \ref{thm:V0monotonicity}]
For $t=1,\dots,T$, set 
\begin{align*}
c_t&=\Var\bigg(\sum_{s=1}^TX_s\bigg)^{-1}\bigg(\Var\bigg(\sum_{s=1}^TX_s\mid \calF_{t-1}\bigg)-\Var\bigg(\sum_{s=1}^TX_s\mid \calF_t\bigg)\bigg),\\
d_t&=\Var\bigg(\sum_{s=1}^TX_s\bigg)^{-1}\bigg(\Var\bigg(\sum_{s=1}^TX_s\mid \calG_{t-1}\bigg)-\Var\bigg(\sum_{s=1}^TX_s\mid \calG_t\bigg)\bigg). 
\end{align*}
By construction, $\sum_{t=1}^Tc_t=\sum_{t=1}^Td_t=1$. Since $\calF_t\subseteq \calG_t$ for every $t$, $\sum_{s=t+1}^Tc_s\geq \sum_{s=t+1}^Td_s$ which is equivalent to $\sum_{s=1}^t d_s \geq \sum_{s=1}^t c_s$ for every $t$ since $\sum_{t=1}^Tc_t=\sum_{t=1}^Td_t$. Applying Lemma \ref{lem:V0monotonicity} completes the proof.   
\end{proof}

\section*{Acknowledgements}
F.~Lindskog would like to acknowledge financial support from the Swedish Research Council, Project 2020-05065, and from L\"ansf\"ors\"akringars Forskningsfond, Project P9.20.

\end{document}